\documentclass[12pt]{iopart}
\usepackage{iopams,mathptm,times,graphicx}
\usepackage{ulem}

\makeatletter
\newcounter{theorem}
\@addtoreset{theorem}{section}
\renewcommand\thetheorem{\arabic{section}.\arabic{theorem}}

\newenvironment{lemma}{\par\medskip\noindent\begingroup{\bf Lemma
             \stepcounter{theorem}\thetheorem.}\ \itshape
             \def\@currentlabel{\thetheorem}}{\endgroup\par\medskip}
\newenvironment{theorem}{\par\medskip\noindent\begingroup{\bf Theorem
             \stepcounter{theorem}\thetheorem.}\ \itshape
             \def\@currentlabel{\thetheorem}}{\endgroup\par\medskip}

\newenvironment{remark}{\par\medskip\noindent\begingroup{\bf Remark
             \stepcounter{theorem}\thetheorem.}\
             \def\@currentlabel{\thetheorem}}{\endgroup\par\medskip}

\newenvironment{proof}{\par\noindent{\bf Proof.} }{\proofbox\par\medskip}

\def\proofbox{\hfill{\ensuremath\Box}}
\newdimen\LENB \newdimen\LENW \newdimen\THI
\newdimen\LENWH \newdimen\LENTOT \newcount\N
\def\vbrknlnele#1#2#3{
  \LENB=#1pt \LENW=#2pt \THI=#3pt
  \LENWH=\LENW \divide\LENWH by 2
  \LENTOT=\LENB \advance\LENTOT by \LENW
  \vbox to \LENTOT{
    \vbox to \LENWH{}
    \nointerlineskip
    \vbox to \LENB{\hbox to \THI{\vrule width \THI height \LENB}}
    \nointerlineskip
    \vbox to \LENWH{}
  }}
\def\vbrknln#1{
  \N=#1
  \vcenter{
    \vbox{
      \loop\ifnum\N>0
        \vbox to 4pt{\vbrknlnele{2}{2}{0.1}}
        \nointerlineskip
        \advance\N by -1
      \repeat
  }}}

\def\hbrknlnele#1#2#3{
  \LENB=#1pt \LENW=#2pt \THI=#3pt
  \LENTOT=\LENB \advance\LENTOT by \LENW
  \vcenter{
    \vbox to \THI{
      \hbox to \LENTOT{
        \hfil
        \vrule width \LENB height \THI
        \hfil}
  }}}

\makeatletter

\def\journal#1&#2,{\begingroup \let\journal=\dummyjournal
               \it #1\unskip~\bf\ignorespaces #2\rm,\endgroup}
\def\dummyjournal{\errmessage{Reference foul up: nested \journal macros}}
\eqnobysec

\def\eqref#1{(\ref{#1})}
\begin{document}

\title[An integrable semi-discrete Degasperis-Procesi equation]{An
integrable semi-discrete Degasperis-Procesi equation}
\author{Bao-Feng Feng$^1$, Ken-ichi Maruno$^2$ and Yasuhiro Ohta$^{3}$ }
\date{\today}

\begin{abstract}
Based on our previous work to the Degasperis-Procesi equation (J. Phys. A {\bf 46} 045205) and the integrable semi-discrete analogue of its short wave limit (J. Phys. A {\bf 48} 135203), we derive an integrable semi-discrete Degasperis-Procesi equation by Hirota's bilinear method. Meanwhile, $N$-soliton solution to the semi-discrete Degasperis-Procesi equation is provided and proved. It is shown that the proposed semi-discrete Degasperis-Procesi equation, along with its $N$-soliton solution converge to ones of the original Degasperis-Procesi equation in the continuous limit.
\kern\bigskipamount\noindent
\end{abstract}


\address{$^1$~School of Mathematical and Statistical Sciences,
The University of Texas-Rio Grade Valley,
Edinburg, TX 78539-2999
}

\address{$^2$~Department of Applied Mathematics,
Waseda University, Tokyo 169-8050, Japan
}
\address{$^3$~Department of Mathematics,
Kobe University, Rokko, Kobe 657-8501, Japan
}
\eads{\mailto{baofeng.feng@utrgv.edu}, \mailto{kmaruno@waseda.jp} and
\mailto{ohta@math.kobe-u.ac.jp}}

\kern-\bigskipamount


\section{Introduction}
In this paper, we are concerned with integrable discretization of the Degasperis-Procesi (DP) equation~\cite{DP,DHH}
\begin{equation}
m_t+3mu_x+m_xu=0\,, \quad m=-a+u-u_{xx}\,,\label{DP-eq1}
\end{equation}
or
\begin{equation}
u_t-3a u_x- u_{txx}+4uu_x=3u_xu_{xx}+uu_{xxx}\,.\label{DP-eq}
\end{equation}

Prior to the DP equation, the Camassa-Holm (CH) equation firstly appeared in a mathematical
search of recursion operators connected with the integrable partial differential
equations \cite{Fokas81} and then has attracted considerable attention
since it was derived as a model equation for shallow water waves \cite{CH93}.
The CH and the DP equations are the only two integrable equations among the b-family equations \cite{Holm03}
\begin{equation}
m_t+bmu_x+m_xu=0\,, \quad m=-a+u-u_{xx}\,,\label{bfamily-eq1}
\end{equation}
or alternatively
\begin{equation}
u_t-ba u_x- u_{txx}+(b+1)uu_x=bu_xu_{xx}+uu_{xxx}\,,\label{bfamily-eq}
\end{equation}
with $b = 2$ and $b = 3$, respectively. \\
It is shown in \cite{Constantin} that both the CH and the DP equations can be used as models
for the propagation of shallow water waves over a flat bed, which
accommodate wave breaking phenomena. These two equations are very similar, only differing by coefficients.
They also share some common properties, for example, when $a=0$, both the CH and DP
equations have multi-peakon solutions \cite{Beals99,Beals00,Lundmark} and admit wave breaking phenomena
\cite{ConstantinEscher,Liu06}.

On the other hand, although the DP equation has an apparent
similarity to the CH equation, there are major structural
differences between these two equations such as the Lax pair, blow
up phenomena and the solutions. The isospectral problem in the Lax
pair for the DP equation is the third-order equation \cite{DHH},
while the isospectral problem for the CH equation is the second
order equation \cite{CH93}. Therefore, the DP equation is much more complicated in term
of the bi-Hamiltonian structures ~\cite{Hone-Wang}, the multi-peakon solution  \cite{Lundmark2},
the inverse scattering transform \cite{Constantin10},
the Riemann-Hilbert problem\cite{MonvelShepelskyDP}, as well as its bilinear equations
and its multi-soliton solutions \cite{Matsuno-DP1,Matsuno-DP2, FMO-DP}.

It is well-known that Hirota's bilinear method is very useful in finding multi-soliton solutions and constructing integrable discretizations to soliton equations \cite{HirotaBook}. Following this direction, the authors have succeeded in constructing integrable discretizations to a class of soliton equations with hodograph transformation. The examples include the short pulse equation and its multi-component generalizations \cite{discreteSP,WKI,discreteCSP,discreteMSP}, the CH equation and its short wave limit (the Hunter-Saxton equation) \cite{dCH,discreteSCH}, the short wave limit of the DP equation (the reduced Ostrovsky equation) \cite{discreteVE}.

The goal of this paper is to construct integrable semi-discretization of the DP equation. It is a challenging problem in compared with the CH equation. In the present paper, we attempt to propose an integrable semi-discrete DP equation based on our previous results regarding the bilinear structure of the DP equation \cite{FMO-DP} and the integrable discretizations of its short wave limit
\cite{discreteVE}.
The remainder of the present paper is organized as follows. In section 2, we gave a brief review for the bilinear equations with their connection with the negative flow of the CKP hierarchy. In section 3, semi-discrete integrable analogues of the bilinear equations of the DP equation are constructed in section 3. Then, we propose an integrable semi-discrete DP equation in section 4. The paper is concluded in section 5 by some comments and further topics.
\section{Review for the bilinear equations of the Degasperis-Procesi
equations and its $N$-soliton solution}
In this section, in order to make this paper self-contained, we will give a brief review to the bilinear equations of
the Degasperis-Procesi equation, as well as its $N$-soliton solution as
presented in \cite{FMO-DP}.
\par
We start with a sequence of $2N \times 2N$ Gram-type determinants \cite{Ohta-Satsuma}
\begin{equation}\label{CKP_det}
  F_{n}=\det_{1\leq i,j\leq 2N}\Big(m_{ij}(n)\Big)\,,
\end{equation}
where the entries of the determinant are defined by
\[
m_{ij}(n)=\delta _{j,2N+1-i}\frac{a^{2}}{2}\frac{\epsilon _{ij}(p_{i}-p_{j})%
}{(1+ap_{i})(1+ap_{j})}+\frac{1}{p_{i}+p_{j}}\left( -\frac{p_{i}}{p_{j}}%
\right) ^{n}\varphi _{i}(0)\varphi _{j}(0),
\]%
with
\[
\varphi _{i}(k)=\left( \frac{1+ap_{i}}{1-ap_{i}}\right) ^{k}e^{\xi
_{i}},\quad \xi _{i}=p^{-1}_{i}s+p_{i}y+\xi _{i0}\,, \quad \epsilon_{ij}=%
\cases{
1 &$i<j$ \cr
-1 &$i>j$}\,.
\]
Meanwhile a sequence of pfaffians which belongs to BKP hierarchy \cite{Hirota-BKP1,Ohta-bookchapter} can be defined by
\begin{equation}\label{BKP_pf}
\tau_{k}=\mathrm{Pf}(1,2,\cdots ,2N)_{k}\,
\end{equation}
whose elements are
\[
(i,j)_{k}=\delta _{j,2N+1-i}\epsilon_{ij}+\frac{p_{i}-p_{j}}{p_{i}+p_{j}}\varphi
_{i}(k)\varphi _{j}(k)\,.
\]
Imposing a reduction condition
\[ p_i^2-p_ip_{2N+1-i}+p_{2N+1-i}^2=\frac{1}{a^2},
\]
and setting
\begin{equation}\label{DP_tau}
f=\tau _{0}\,,\qquad g=\tau _{1}\,, \qquad F=F_{0}\,, \qquad G=F_{1}\,,
\end{equation}
as shown in \cite{FMO-DP}, the following relations hold between the determinants and pfaffians
\begin{eqnarray}
&&\left(D_{y}D_{s}-aD_{y}-\frac{1}{a}D_{s}\right)g\cdot f=0\,, \label{blinear_DP1}\\
&&gf=cG\,, \label{blinear_DP2} \\
&&(-aD_{y}+1)g\cdot f=cF\,, \label{blinear_DP3} \\
&&\left(\frac{1}{2}D_{y}D_{s}-1\right)F\cdot F=-G^2\,, \label{blinear_DP4}
\end{eqnarray}
where
\[
c=\prod_{i=1}^{2N}2p_{i}\frac{1+ap_{i}}{1-ap_{i}}\,,
\]
and $D_{y}D_{s}$ is the Hirota $D$-operator defined by
\[
D_{y}^{m}D_{s}^{n}f(y,s)\cdot g(y,s)=\left( \frac{\partial }{\partial y}-%
\frac{\partial }{\partial y{^{\prime }}}\right)^{m}\left( \frac{\partial }{%
\partial s}-\frac{\partial }{\partial s{^{\prime }}}\right)
^{n}f(y,s)g(y^{\prime },s^{\prime })|_{y=y^{\prime },s=s^{\prime }}\,.
\]
We comment here that the above four equations (\ref{blinear_DP1})-- (\ref{blinear_DP4}) are basically equivalent to Eqs. (2.20)--(2.24) in \cite{FMO-DP} although they seem  slightly different. In what follows, we will show briefly how the equations (\ref{blinear_DP1})-- (\ref{blinear_DP4}) yield the DP equation through a dependent variable transformation
\begin{equation}\label{tranf-u}
  u=\left(\ln \frac{g}{f}\right)_{s}\,,
\end{equation}
and a hodograph transformation
\begin{equation}\label{tranf-hodograph}
x=-\frac{1}{a}y+\ln \frac{g}{f}\,,\qquad t=s\,.
\end{equation}
Dividing $gf$ on both sides, the first three equations (\ref{blinear_DP1})-- (\ref{blinear_DP3}) can be rewritten as
\begin{equation}
(\ln gf)_{ys}+\left(\left(\ln \frac{g}{f}\right)_{y}-\frac{1}{a}\right)
\left(\left(\ln \frac{g}{f} \right)_{s}-a\right)-1=0\,,  \label{DP1n}
\end{equation}
\begin{equation}
1=\frac{cG}{{g}{f}}\,,  \label{DP2n}
\end{equation}
\begin{equation}
-a\left(\ln \frac{g}{f}\right)_{y}+1=\frac{cF}{{g}{f}}\,.  \label{DP3n}
\end{equation}%
While, by dividing $F^{2}$ on both sides, the bilinear equation (\ref{blinear_DP4}) becomes
\begin{equation}
(\ln F)_{ys}-1=-\frac{G^2}{F^2}\,.  \label{DP4n}
\end{equation}
With the use of (\ref{DP2n}), (\ref{DP3n}) becomes
\begin{equation}
-a\left(\ln \frac{g}{f}\right)_{y}+1=\frac{F}{G}\,.  \label{DP6n}
\end{equation}
Subtracting (\ref{DP4n}) from (\ref{DP1n}), one obtains
\begin{equation}
\left(\ln \frac{G}{F}\right)_{ys}+\left(\left(\ln \frac{g}{f}\right)_{y}-\frac{1}{a}\right)
\left(\left(\ln \frac{g}{f} \right)_{s}-a\right)=\frac{G^2}{F^2}  \label{DP5n}
\end{equation}
by referring to (\ref{DP2n}).

Introducing an intermediate variable $\rho=G/F$, one can calculate that
\[
\frac{\partial x}{\partial y}=-\frac{1}{a}+\left(\ln \frac{g}{f}\right)_{y}=-\frac{1}{%
a\rho},\qquad \frac{\partial x}{\partial s}=\left(\ln \frac{g}{f}\right)_{s}=u,
\]
based on the transformations (\ref{tranf-u})--(\ref{tranf-hodograph}), which yields a conversion formula
\begin{equation}\label{conversion}
\partial _{y}=-\frac{1}{a\rho}\partial _{x},\qquad \partial _{s}=\partial
_{t}+u\partial _{x}.
\end{equation}
Substituting (\ref{DP6n}) into (\ref{DP5n}), one obtains
\begin{equation}
(\ln \rho)_{ys}-\frac{1}{a\rho}(u-a)=\rho^{2}\,,  \label{DP5p}
\end{equation}
which can be rewritten as%
\begin{equation}
((\ln \rho)_{s})_{x}+u-a=-a\rho^{3}\,.  \label{DP7n}
\end{equation}
On the other hand, differentiating (\ref{DP6n}) with respect to $s$, one yields
\begin{equation}
\left(\frac{1}{\rho}\right)_{s}+ au_{y}=0\,,  \label{DP6p}
\end{equation}
which, in turn, becomes
\begin{equation}
(\ln \rho)_{s}=-u_{x}  \label{DP8n}
\end{equation}
by using the conversion formula (\ref{conversion}).

In the last, eliminating $\rho$ from (\ref{DP7n})--(\ref{DP8n}), one obtains
\begin{equation}\label{DP-alt}
(\partial_{t}+u\partial _{x})\ln (u-u_{xx}-a)=-3u_{x}\,,
\end{equation}
which is nothing but the Degasperis-Procesi equation (\ref{DP-eq}).
\section{Semi-discrete analogue of equations (\ref{blinear_DP1})--(\ref{blinear_DP4})}
Based on the results mentioned in the previous section, we attempt to construct an integrable semi-discrete analogue of the
DP equation (\ref{DP-eq}) by using Hirota's bilinear method. The key point is how to obtain discrete analogues of the equations (\ref{blinear_DP1})--(\ref{blinear_DP1}) including the bilinear equations possessing  $N$-soliton solutions.

Keeping in mind that the Degasperis-Procesi equation is derived from
pseudo 3-reduction of the CKP hierarchy, we start with the Gram-type determinants, which are soliton solutions of the CKP hierarchy,
\[
F_{k,l}=\det_{1\leq i,j\leq 2N}\Big(m_{ij}(k,l)\Big),\qquad
G_{k,l}=\det_{1\leq i,j\leq 2N}\Big(m_{ij}^{\prime }(k,l)\Big)\,,
\]%
where
\[
m_{ij}(k,l)=C_{ij}+\frac{1}{p_{i}+p_{j}}\varphi _{i}^{(0)}(k,l)\varphi
_{j}^{(0)}(k,l)\,,
\]%
\[
m_{ij}^{\prime }(k,l)=C_{ij}+\frac{1}{p_{i}+p_{j}}\left(-\frac{p_{j}}{p_{i}}\right)\frac{%
1+bp_{i}}{1-bp_{j}}\varphi _{i}^{(0)}(k,l)\varphi _{j}^{(0)}(k,l)\,,
\]%
with
\[
C_{ij}=C_{ji},\quad \varphi _{i}^{(n)}(k,l)=p_{i}^{n}\left( \frac{1+ap_{i}}{%
1-ap_{i}}\right) ^{k}\left( \frac{1+bp_{i}}{1-bp_{i}}\right) ^{l}e^{\xi
_{i}},\quad \xi _{i}=p^{-1}_{i}s +\xi _{i0}\,.
\]%
Here $2b$ (not $b$) is the mesh size in $y$-direction. The relation between $F_{k,l}$ and $G_{k,l}$ is given  by the following lemma.
\begin{lemma}
\begin{equation}
(D_{s}-2b)F_{k,l+1}\cdot F_{k,l}=-2bG_{k,l}^{2}\,,  \label{Bilinear1}
\end{equation}
\end{lemma}

\begin{proof}
It can be easily verified that
\[
\partial_s m_{ij}(k,l)=\varphi_i^{(-1)}(k,l)\varphi_j^{(-1)}(k,l)\,,
\]
\[
m_{ij}(k,l+1)=m_{ij}(k,l) +\frac{2b}{(1-bp_i)(1-bp_j)}\varphi_i^{(0)}(k,l)%
\varphi_j^{(0)}(k,l)\,,
\]
and
\[
m_{ij}^{\prime }(k,l)=m_{ij}(k,l) -\frac{1}{1-bp_j}\varphi_i^{(-1)}(k,l)%
\varphi_j^{(0)}(k,l)\,.
\]
Then we have
\[
\partial_sF_{k,l}=\left|%
\matrix{
m_{ij}(k,l) &\varphi_i^{(-1)}(k,l) \cr
-\varphi_j^{(-1)}(k,l) &0}\right|\,,
\]
\[
F_{k,l+1}=\left|%
\matrix{
m_{ij}(k,l) &\displaystyle\frac{2b}{1-bp_i}\varphi_i^{(0)}(k,l) \cr
\displaystyle -\frac{1}{1-bp_j}\varphi_j^{(0)}(k,l) &1}\right|\,,
\]
\[
G_{k,l}=\left|%
\matrix{
m_{ij}(k,l) &\varphi_i^{(-1)}(k,l) \cr
\displaystyle\frac{1}{1-bp_j}\varphi_j^{(0)}(k,l) &1}\right| =\left|%
\matrix{
m_{ij}(k,l) &\displaystyle\frac{1}{1-bp_i}\varphi_i^{(0)}(k,l) \cr
\varphi_j^{(-1)}(k,l) &1}\right|\,,
\]
\[
(\partial_s-2b)F_{k,l+1} =\left|%
\matrix{
m_{ij}(k,l) &\varphi_i^{(-1)}(k,l)
&\displaystyle\frac{2b}{1-bp_i}\varphi_i^{(0)}(k,l) \cr
-\varphi_j^{(-1)}(k,l) &0 &-2b \cr
\displaystyle -\frac{1}{1-bp_j}\varphi_j^{(0)}(k,l) &-1 &1}\right|\,.
\]
By the Jacobi identity of determinants, we obtain
\[
(\partial_s-2b)F_{k,l+1}\times F_{k,l}
=F_{k,l+1}\times\partial_s F_{k,l}-(-2bG_{k,l})\times(-G_{k,l})\,,
\]
which is exactly the bilinear equation (\ref{Bilinear1}).
\end{proof}
Next, we perform reductions similar to the pseudo 3-reduction of the CKP
hierarchy in the continuous case. To this end, we take
\begin{equation}
C_{ij}=\delta_{j,2N+1-i}c_i, \qquad c_{2N+1-i}=c_i\,,
\end{equation}
and further assume
\begin{equation}
c_{ij}={2p_i}C_{ij} \frac{1+ap_i}{1-ap_j} \frac{1-bp_j}{1+bp_i}\,.
\end{equation}
By imposing a reduction condition
\begin{equation}
\frac{(1-a^2p_{2N+1-i}^2)(1-b^2p_i^2)}{p_i} +\frac{%
(1-a^2p_i^2)(1-b^2p_{2N+1-i}^2)}{p_{2N+1-i}}=0\,,
\end{equation}
or, equivalently,
\begin{equation}
\frac{p_i(1+ap_i)(1-bp_{2N+1-i})}{(1-ap_{2N+1-i})(1+bp_i)} =-\frac{%
p_{2N+1-i}(1+ap_{2N+1-i})(1-bp_i)}{(1-ap_i)(1+bp_{2N+1-i})}\,,
\end{equation}
it then follows that
\begin{eqnarray*}
c_{ij}&=&\delta_{j,2N+1-i} c_i \frac{{2p_i}(1+ap_i)}{1-ap_{2N+1-i}} \frac{%
1-bp_{2N+1-i}}{1+bp_i} \\
&=&-\delta_{j,2N+1-i}c_{2N+1-i} \frac{2p_{2N+1-i}(1+ap_{2N+1-i})} {%
1-ap_i} \frac{1-bp_i}{1+bp_{2N+1-i}} \\
&=&-c_{ji}\,.
\end{eqnarray*}
Therefore, we can define a pfaffian of the form
\[
f_{kl}= \mathrm{Pf} (1,2,\cdots,2N)_{kl}\,,
\]
whose elements defined by
\[
(i,j)_{kl}=c_{ij} +\frac{p_i-p_j}{p_i+p_j}\varphi_i^{(0)}(k,l)%
\varphi_j^{(0)}(k,l)\,.
\]
Then the following lemma provides a bilinear equation satisfied by the pfaffian $f_{kl}$.
\begin{lemma}
\begin{equation}  \label{Bilinear2}
\left(\frac{1}{a+b}D_s-1\right)f_{k+1,l+1}\cdot f_{kl} =\left(\frac{1}{a-b}%
D_s-1\right)f_{k+1,l}\cdot f_{k,l+1}\,,
\end{equation}
\end{lemma}
\begin{proof}
Letting
\[
(i,d_n)_{kl}=\varphi_i^{(n)}(k,l)\,, \quad (d_m,d_n)_{kl}=0\,,
\]
\[
(i,d^k)_{kl}=\varphi_i^{(0)}(k+1,l)\,, \quad (d_0,d^k)_{kl}=1\,, \quad
(d_{-1},d^k)_{kl}=-a\,,
\]
\[
(i,d^l)_{kl}=\varphi_i^{(0)}(k,l+1)\,, \quad (d_0,d^l)_{kl}=1\,, \quad
(d_{-1},d^l)_{kl}=-b\,,
\]
and
\[
(d^l,d^k)_{kl}=\frac{a-b}{a+b}\,,
\]
it is shown in Appendix that
\begin{equation}\label{pf-relation1}
\partial_s f_{kl} =(1,2,\cdots,2N,d_{-1},d_0)_{kl}\,,
\end{equation}
\begin{equation}\label{pf-relation2}
f_{k+1,l}=(1,2,\cdots,2N,d_0,d^k)_{kl}\,,
\end{equation}
\begin{equation}\label{pf-relation3}
f_{k,l+1}=(1,2,\cdots,2N,d_0,d^l)_{kl}\,,
\end{equation}
\begin{equation}\label{pf-relation4}
(\partial_s-a)f_{k+1,l} = (1,2,\cdots,2N,d_{-1},d^k)_{kl}\,,
\end{equation}
\begin{equation}\label{pf-relation5}
(\partial_s-b)f_{k,l+1} = (1,2,\cdots,2N,d_{-1},d^l)_{kl}\,,
\end{equation}
\begin{equation}\label{pf-relation6}
\frac{a-b}{a+b}f_{k+1,l+1} = (1,2,\cdots,2N,d^l,d^k)_{kl}\,,
\end{equation}
\begin{equation}\label{pf-relation7}
(\partial_s-a-b)\frac{a-b}{a+b}f_{k+1,l+1} = (1,2,\cdots,2N,d_{-1},d_0,
d^l,d^k)_{kl}\,.
\end{equation}
Therefore, an algebraic identity of pfaffian \cite{HirotaBook}
\begin{eqnarray*}
  && {\rm Pf} (\cdots, d_{-1}, d_{0}, d^l, d^k) {\rm Pf} (\cdots)= {\rm Pf} (\cdots, d_{-1}, d_{0}) {\rm Pf} (\cdots, d^l, d^k) \\
   && \quad - {\rm Pf} (\cdots, d_{-1},  d^l) {\rm Pf} (\cdots, d_{0}, d^k) +
   {\rm Pf} (\cdots, d_{-1}, d^k) {\rm Pf} (\cdots, d_0, d^l)\,,
\end{eqnarray*}
implies
\begin{eqnarray*}
&& (\partial_s-a-b)\frac{a-b}{a+b}f_{k+1,l+1}\times f_{kl} \\
= && \frac{a-b}{a+b}f_{k+1,l+1}\times\partial_sf_{kl}
-f_{k+1,l}\times(\partial_s-b)f_{k,l+1} +(\partial_s-a)f_{k+1,l}\times
f_{k,l+1}\,,
\end{eqnarray*}
which is nothing but the bilinear equation (\ref{Bilinear2}).
\end{proof}
The lemma below states the relations between pfaffian and determinant defined previously
\begin{lemma}
\begin{equation}  \label{Bilinear3}
f_{k+1,l}f_{kl}=c'G_{kl}\,,
\end{equation}
\begin{equation}  \label{Bilinear4}
(a-b)f_{k+1,l+1}f_{kl}-(a+b)f_{k+1,l}f_{k,l+1}=-2bc'F_{k,l+1}\,,
\end{equation}
with
\[
c'=\prod_{i=1}^{2N}2p_i\frac{1+ap_i}{1-ap_i}\frac{1-bp_i}{1+bp_i}\,.
\]
\end{lemma}

\begin{proof}
By an identity of pfaffian \cite{FMO-DP}
\begin{eqnarray}
\fl &&f_{k+1,l}f_{kl}=\mathrm{Pf}\pmatrix{ \matrix{(i,j)_{kl} \cr {}}
&\varphi_i^{(0)}(k,l) &\varphi_i^{(0)}(k+1,l) \cr &&1} \times\mathrm{Pf}\Big(%
(i,j)_{kl}\Big) \nonumber \\
\fl &&\qquad =\left|\matrix{ (i,j)_{kl} &\varphi_i^{(0)}(k,l) \cr
\varphi_j^{(0)}(k+1,l) &1}\right| \nonumber \\
\fl &&\qquad =\det_{1\le i,j\le 2N}\Big( (i,j)_{kl}-\varphi_i^{(0)}(k,l)%
\varphi_j^{(0)}(k+1,l)\Big) \nonumber \\
\fl &&\qquad =\det\Big( c_{ij}+\frac{p_i-p_j}{p_i+p_j}\varphi_i^{(0)}(k,l)%
\varphi_j^{(0)}(k,l) -\varphi_i^{(0)}(k,l)\varphi_j^{(0)}(k+1,l)\Big) \nonumber \\
\fl &&\qquad =\det\Big( c_{ij}+\frac{-2p_j}{p_i+p_j}\frac{1+ap_i}{1-ap_j}
\varphi_i^{(0)}(k,l)\varphi_j^{(0)}(k,l)\Big) \nonumber \\
\fl &&\qquad =c' \det\Big(c_{ij} \frac{1}{2p_i}\frac{1-ap_j}{1+ap_i}\frac{1+bp_i}{%
1-bp_j} +\frac{1}{p_i+p_j}\left(-\frac{p_j}{p_i}\right)\frac{1+bp_i}{1-bp_j}
\varphi_i^{(0)}(k,l)\varphi_j^{(0)}(k,l)\Big) \nonumber \\
\fl &&\qquad =c'\det\Big(C_{ij} +\frac{1}{p_i+p_j}\left(-\frac{p_j}{p_i}\right)\frac{1+bp_i}{%
1-bp_j} \varphi_i^{(0)}(k,l)\varphi_j^{(0)}(k,l)\Big) \nonumber \\
\fl &&\qquad =c'\det\Big(m_{ij}^{\prime }(k,l)\Big)\,.
\end{eqnarray}
Thus, Eq. (\ref{Bilinear3}) is proved. Next, we proceed to the proof of Eq. (\ref{Bilinear4}). Firstly,
by the same identity as above, the products of pfaffians can be rewritten into determinants
\begin{eqnarray}
\fl &&\frac{a-b}{a+b}f_{k+1,l+1}f_{kl} =\mathrm{Pf}\pmatrix{ \matrix{(i,j)_{kl}
\cr {}} &\varphi_i^{(0)}(k,l+1) &\varphi_i^{(0)}(k+1,l) \cr
&&\displaystyle\frac{a-b}{a+b}} \times\mathrm{Pf}\Big((i,j)_{kl}\Big) \nonumber \\
\fl &&\qquad =\left|\matrix{ (i,j)_{kl} &\varphi_i^{(0)}(k,l+1) \cr
\varphi_j^{(0)}(k+1,l) &\displaystyle\frac{a-b}{a+b}}\right|\,,
\end{eqnarray}
\begin{eqnarray}
\fl && f_{k+1,l}f_{k,l+1} \nonumber \\
\fl && =\mathrm{Pf}\pmatrix{ \matrix{(i,j)_{kl} \cr {}} &\varphi_i^{(0)}(k,l)
&\varphi_i^{(0)}(k+1,l) \cr &&1} \times\mathrm{Pf}\pmatrix{
\matrix{(i,j)_{kl} \cr {}} &\varphi_i^{(0)}(k,l) &\varphi_i^{(0)}(k,l+1) \cr
&&1} \nonumber \\
\fl && =\left|\matrix{ (i,j)_{kl} &\varphi_i^{(0)}(k,l) &\varphi_i^{(0)}(k,l+1)
\cr -\varphi_j^{(0)}(k,l) &0 &1 \cr -\varphi_j^{(0)}(k+1,l) &-1
&\displaystyle -\frac{a-b}{a+b}}\right|\,.
\end{eqnarray}
Consequently,
\begin{eqnarray}
\fl &&\frac{a-b}{a+b}f_{k+1,l+1}f_{kl}-f_{k+1,l}f_{k,l+1} \nonumber \\
\fl && =\left|\matrix{ (i,j)_{kl} &\varphi_i^{(0)}(k,l+1) \cr
\varphi_j^{(0)}(k+1,l) &\displaystyle\frac{a-b}{a+b}}\right| -\left|\matrix{
(i,j)_{kl} &\varphi_i^{(0)}(k,l) &\varphi_i^{(0)}(k,l+1) \cr
-\varphi_j^{(0)}(k,l) &0 &1 \cr -\varphi_j^{(0)}(k+1,l) &-1 &\displaystyle
-\frac{a-b}{a+b}}\right| \nonumber \\
\fl && =\left|\matrix{ (i,j)_{kl} &\varphi_i^{(0)}(k,l) &\varphi_i^{(0)}(k,l+1)
\cr -\varphi_j^{(0)}(k,l) &1 &1 \cr \varphi_j^{(0)}(k+1,l) &1
&\displaystyle\frac{a-b}{a+b}}\right| \nonumber \\
\fl && =\left|\matrix{ (i,j)_{kl}+\varphi_i^{(0)}(k,l)\varphi_j^{(0)}(k,l)
&\varphi_i^{(0)}(k,l) &\varphi_i^{(0)}(k,l+1)-\varphi_i^{(0)}(k,l) \cr 0 &1
&0 \cr \varphi_j^{(0)}(k+1,l)+\varphi_j^{(0)}(k,l) &1
&\displaystyle\frac{-2b}{a+b}}\right| \nonumber \\
\fl && =\left|\matrix{ (i,j)_{kl}+\varphi_i^{(0)}(k,l)\varphi_j^{(0)}(k,l)
&\varphi_i^{(0)}(k,l+1)-\varphi_i^{(0)}(k,l) \cr
\varphi_j^{(0)}(k+1,l)+\varphi_j^{(0)}(k,l) &\displaystyle\frac{-2b}{a+b}}%
\right| \nonumber \\
\fl &&\ =\frac{-2b}{a+b}\det\Big( c_{ij}+\frac{2p_i(1+ap_i)(1+bp_j)}{%
(p_i+p_j)(1-bp_i)(1-ap_j)} \varphi_i^{(0)}(k,l)\varphi_j^{(0)}(k,l)\Big) \nonumber \\
&& =\frac{-2b}{a+b}\det\Big( c_{ij}+\frac{2p_i(1+ap_i)(1-bp_j)}{%
(p_i+p_j)(1+bp_i)(1-ap_j)} \varphi_i^{(0)}(k,l+1)\varphi_j^{(0)}(k,l+1)\Big)
\nonumber \\
\fl && =\frac{-2bc'}{a+b} \det\Big( c_{ij}\frac{1}{2p_i}\frac{1-ap_j}{1+ap_i}%
\frac{1+bp_i}{1-bp_j} +\frac{1}{p_i+p_j}\varphi_i^{(0)}(k,l+1)%
\varphi_j^{(0)}(k,l+1)\Big) \nonumber \\
&& =\frac{-2bc'}{a+b}\det\Big(m_{ij}(k,l+1)\Big)\,.
\end{eqnarray}
Multiplying both sides by $(a+b)$, we arrive at (\ref{Bilinear4}).
\end{proof}
Summarizing what we have discussed and letting $f_l=f_{0l}$, $g_l=f_{1l}$, $F_l=F_{0l}$, $G_l=G_{0l}$, we arrive at the
following four equations
\begin{equation}\label{Bilinear-sd1}
    \left(\frac{1}{a+b}D_s-1\right)g_{l+1}\cdot f_l =\left(\frac{1}{a-b}D_s-1\right)g_l\cdot f_{l+1}\,,
\end{equation}
\begin{equation}\label{Bilinear-sd2}
    g_lf_l=c'G_l\,,
\end{equation}
\begin{equation}\label{Bilinear-sd3}
    (D_s-2b)F_{l+1}\cdot F_l=-2bG^2_l\,,
\end{equation}
\begin{equation}\label{Bilinear-sd4}
(a-b)g_{l+1}f_l-(a+b)g_lf_{l+1}=-2bc'F_{l+1}\,.
\end{equation}
In fact, Eqs. (\ref{Bilinear-sd1})--(\ref{Bilinear-sd4}) are integrable semi-discrete analogues of Eqs. (\ref{blinear_DP1})--(\ref{blinear_DP4}). In other words, in the limit of $b \to 0$, Eqs. (\ref{Bilinear-sd1})--(\ref{Bilinear-sd4}) converge to Eqs. (\ref{blinear_DP1})--(\ref{blinear_DP4}), respectively. Meanwhile, the pfaffian and determinant solutions satisfying Eqs. (\ref{Bilinear-sd1})--(\ref{Bilinear-sd4}) also converge to the  pfaffian and determinant solutions satisfying Eqs. (\ref{blinear_DP1})--(\ref{blinear_DP4}).

Note that $2b$ is the mesh size. In the limit of $b \to 0$, we have $c' \to c$,
$$
f_{l} \to f, \quad g_{l} \to g\,, \quad f_{l+1} \to f + 2b f_{y}\,,  \quad g_{l+1} \to g + 2b g_{y}\,,
$$
and similar relations for the determinants $F_l$, $G_l$, $F_{l+1}$ and $G_{l+1}$. Obviously, (\ref{Bilinear-sd2}) goes to (\ref{blinear_DP2}) as $b \to 0$. It can be easily shown that
$$
\frac{1}{2b} D_s F_{l+1}\cdot F_{l} \to \frac 12 D_sD_y F \cdot F\,,
$$
and
$$
\frac{1}{2b} (g_{l+1}f_l-g_lf_{l+1}) \to D_y g \cdot f\,, \qquad \frac{1}{2} (g_{l+1}f_l+g_lf_{l+1}) \to gf\,.
$$
Therefore, by dividing $2b$ on both sides, (\ref{Bilinear-sd3}) and (\ref{Bilinear-sd4}) converge to (\ref{blinear_DP3}) and (\ref{blinear_DP4}), respectively, as $b \to 0$.
Now we show the convergence of the first bilinear equation. It is obvious by noting that
\begin{eqnarray}
&&\frac{a}{2b} \left(\frac {1}{a+b} D_s g_{l+1} \cdot  f_{l} - \frac{1}{a-b} D_s g_{l} \cdot f_{l+1} \right) \nonumber \\
&& \quad \to \frac{1}{2b}  \left( 1-\frac{b}{a} \right) D_s g_{l+1} \cdot  f_{l}
- \left( 1+\frac{b}{a} \right) D_s g_{l} \cdot  f_{l+1} \nonumber \\
&& \quad \to \frac{1}{2b}   D_s (g_{l+1} \cdot  f_{l}- g_{l} \cdot f_{l+1})
- \frac {1}{2a}  D_s (g_{l+1} \cdot  f_{l}+ g_{l} \cdot f_{l+1}) \nonumber \\
&& \quad \to \frac{1}{2}   D_s D_y  g \cdot  f - \frac{1}{a} g f \,,
\end{eqnarray}
and
\begin{eqnarray}
&&  \frac{a}{2b}  (g_{l+1} \cdot  f_{l}- g_{l} \cdot f_{l+1}) \nonumber \\
&& \quad \to a    D_y  g \cdot  f \,.
\end{eqnarray}
\section{Semi-discrete Degasperis-Procesi equation}
Now that we have constructed integrable semi-discrete analogues (\ref{Bilinear-sd1})--(\ref{Bilinear-sd4}) of a set of equations (\ref{blinear_DP1})--(\ref{blinear_DP4}) which derive the Degasperis-Procesi equation, we proceed to construct an integrable semi-discrete Degasperis-Procesi equation based on Hirota's bilinear method.
First, let us work on the bilinear equation (\ref{Bilinear-sd1}), which can be recast into
\begin{equation*}
\fl 2g_{l+1}f_l\left((a-b)\left(\ln\frac{g_{l+1}}{f_l}\right)_s-a^2+b^2\right)
-2g_lf_{l+1}\left((a+b)\left(\ln\frac{g_l}{f_{l+1}}\right)_s-a^2+b^2\right)=0\,,
\end{equation*}
by multiplying $2(a^2-b^2)$ on both sides, or,
\begin{eqnarray*}
\fl &&((a-b)g_{l+1}f_l+(a+b)g_lf_{l+1})
\left(\ln\frac{g_{l+1}f_{l+1}}{g_lf_l}\right)_s
\\
\fl &&\qquad+((a-b)g_{l+1}f_l-(a+b)g_lf_{l+1})
\left(\ln\frac{g_{l+1}g_l}{f_{l+1}f_l}\right)_s
-2(a^2-b^2)(g_{l+1}f_l-g_lf_{l+1})=0\,.
\end{eqnarray*}
Rearranging the terms, one obtains
\begin{eqnarray}
\fl &&((a-b)g_{l+1}f_l+(a+b)g_lf_{l+1})
\left(\ln\frac{g_{l+1}f_{l+1}}{g_lf_l}\right)_s -2b((a-b)g_{l+1}f_l+(a+b)g_lf_{l+1})
\nonumber\\
\fl &&\quad +((a-b)g_{l+1}f_l-(a+b)g_lf_{l+1})
\left(\left(\ln\frac{g_{l+1}g_l}{f_{l+1}f_l}\right)_s-2a\right)=0\,.
\label{sdDP1a}
\end{eqnarray}
Dividing $((a-b)g_{l+1}f_l+(a+b)g_lf_{l+1})$ on both sides of (\ref{sdDP1a}), we  have
\begin{equation}
\fl \left(\ln \frac{g_{l+1}f_{l+1}}{g_{l}f_{l}}\right)_{s}+\frac{%
(a-b)g_{l+1}f_{l}-(a+b)g_{l}f_{l+1}}{(a-b)g_{l+1}f_{l}+(a+b)g_{l}f_{l+1}}%
\left(\left(\ln\frac{g_{l+1}g_l}{f_{l+1}f_l}\right)_s-2a\right)-2b=0\,.  \label{sdDP1n}
\end{equation}%
Secondly, Eqs. (\ref{Bilinear-sd3}) and (\ref{Bilinear-sd4}) can be easily rewritten as
\begin{equation}
\frac{(a-b)g_{l+1}f_{l}-(a+b)g_{l}f_{l+1}}{%
(a-b)g_{l+1}f_{l}+(a+b)g_{l}f_{l+1}}=-\frac{2bc'F_{l+1}}{%
(a-b)g_{l+1}f_{l}+(a+b)g_{l}f_{l+1}}\,,  \label{sdDP3n}
\end{equation}%
and
\begin{equation}
\left(\ln \frac{F_{l+1}}{F_{l}}\right)_{s}-2b=-2b\frac{G^2_{l}}{F_{l+1}F_{l}}\,,
\label{sdDP4n}
\end{equation}%
respectively. Subtracting Eq. (\ref{sdDP4n}) from Eq. (\ref{sdDP1n}), we get
\begin{equation}
\fl \left(\ln \frac{G_{l+1}F_{l}}{G_{l}F_{l+1}}\right)_{s}+\frac{%
(a-b)g_{l+1}f_{l}-(a+b)g_{l}f_{l+1}}{(a-b)g_{l+1}f_{l}+(a+b)g_{l}f_{l+1}}%
\left(\left(\ln\frac{g_{l+1}g_l}{f_{l+1}f_l}\right)_s-2a\right)=2b\frac{G^2_{l}}{%
F_{l+1}F_{l}}  \label{sdDP5n}
\end{equation}%
by referring to Eq. (\ref{Bilinear-sd2}).

Introducing variable transformations
\begin{equation}
\label{sd_u_trf}
u_{l}=\left(\ln \frac{g_{l}}{f_{l}}\right)_{s}\,, \quad r_{l}=\frac{G_{l}}{F_{l}},
\end{equation}
and a discrete hodograph transformation
\begin{equation}
\label{sd_hodograph_trf}
\delta _{l}=-\frac{4bc'F_{l+1}}{(a-b)g_{l+1}f_{l}+(a+b)g_{l}f_{l+1}}\,, \quad t=s\,,
\end{equation}  we  then have
\begin{equation}
\frac{(a-b)g_{l+1}f_{l}-(a+b)g_{l}f_{l+1}}{%
(a-b)g_{l+1}f_{l}+(a+b)g_{l}f_{l+1}}=\frac{\delta _{l}}{2}  \label{sdDP2n}
\end{equation}%
from Eq. (\ref{sdDP3n}). A substitution of Eq. (\ref{sdDP2n})  into Eq. (\ref%
{sdDP5n}) leads to
\begin{equation}
\left(\ln \frac{r_{l+1}}{r_{l}}\right)_{s}+\delta _{l}\left(\frac{u_{l+1}+u_{l}}{2}-a\right)=2b%
\frac{F_{l}}{F_{l+1}}r_{l}^{2}\, \label{sdDP5p}
\end{equation}
by referring variable transformations.

Furthermore, based on (\ref{sdDP3n}) and (\ref{sd_hodograph_trf}), we obtain
\begin{equation}
-\frac{a-b}{b}\frac{g_{l+1}f_{l}}{c'F_{l+1}}=\frac{2}{\delta _{l}}+1\,,\qquad -%
\frac{a+b}{b}\frac{g_{l}f_{l+1}}{c'F_{l+1}}=\frac{2}{\delta _{l}}-1\,. \label{sdDP6p}
\end{equation}
Multiplying above two equations leads to
\begin{equation}
\frac{a^{2}-b^{2}}{b^{2}}\frac{F_{l}}{F_{l+1}}r_{l+1}r_{l}=\frac{4}{\delta
_{l}^{2}}-1\,,  \label{sdDP3pa}
\end{equation}%
while dividing them yields
\begin{equation}
\frac{a-b}{a+b}\frac{g_{l+1}f_{l}}{f_{l+1}g_{l}}=\frac{2+\delta _{l}}{%
2-\delta _{l}}\,.  \label{sdDP3pb}
\end{equation}
Taking logarithmic differentiation of Eq. (\ref{sdDP3pa}) and Eq. (\ref%
{sdDP3pb}) with respect to $s$, one obtains
\begin{equation}
(\ln r_{l+1}r_{l})_{s}-\left(\ln \frac{F_{l+1}}{F_{l}}\right)_{s}=\frac{-8}{(4-\delta
_{l}^{2})\delta _{l}}\frac{ d\,\delta _{l}}{d\,s}\,,  \label{sdDP4pa}
\end{equation}%
and
\begin{equation}
u_{l+1}-u_{l}=\frac{4}{4-\delta _{l}^{2}} \frac{ d\,\delta _{l}}{d\,s}\,, \,,
\label{sdDP4pb}
\end{equation}%
respectively. Eq. (\ref{sdDP4pb}) can be rewritten as
\begin{equation}
\frac{ d\,\delta _{l}}{d\,s}= \left( 1- \frac{\delta _{l}^{2}}{4} \right)(u_{l+1}-u_{l}) \,
\label{sdDP4pc}
\end{equation}%
which constitutes one of the semi-discrete DP equation, describing the time evolution of the nonuniform mesh.
Eliminating $d\delta_{l}/{ds}$ from Eqs. (\ref{sdDP4pa}) and (\ref{sdDP4pb}), one obtains
\begin{eqnarray}
\frac{u_{l+1}-u_{l}}{\delta _{l}} &=&-\frac{1}{2}(\ln r_{l+1}r_{l})_{s}+%
\frac{1}{2}\left(\ln \frac{F_{l+1}}{F_{l}}\right)_{s} \nonumber \\
&=&-\frac{1}{2}(\ln r_{l+1}r_{l})_{s}+b-b\frac{G^2_{l}}{F_{l+1}F_{l}} \nonumber  \\
&=&-\frac{1}{2}(\ln r_{l+1}r_{l})_{s}+b-b\frac{F_{l}}{F_{l+1}}r_{l}^{2} \nonumber \\
&=&-\frac{1}{2}(\ln r_{l+1}r_{l})_{s}+b-\frac{1}{2}\left(\ln \frac{r_{l+1}}{%
r_{l}}\right)_{s}-\frac{\delta _{l}}{2}\left(\frac{u_{l+1}+u_{l}}{2}-a\right)  \nonumber \\
&=&-(\ln r_{l+1})_{s}+b-\frac{\delta _{l}}{2}\left(\frac{u_{l+1}+u_{l}}{2}-a\right)\,.
\end{eqnarray}%
Here Eqs. (\ref{sdDP4n}) and (\ref{sdDP5p}) are used.


In summary, we propose an integrable semi-discrete Degasperis-Procesi equation
\begin{eqnarray}
&&\frac{1}{\delta _{l}}\left(\ln \frac{r_{l+1}}{r_{l}}\right)_{s}+\frac{u_{l+1}+u_{l}}{%
2}-a=\frac{2b}{\delta _{l}}\frac{r_{l}}{r_{l+1}}\frac{\frac{4}{\delta
_{l}^{2}}-1}{\frac{a^{2}}{b^{2}}-1}\,,  \label{sdDPf1} \\
&&(\ln r_{l+1})_{s} = -\frac{u_{l+1}-u_{l}}{\delta _{l}}+b-\frac{\delta _{l}}{2}\left(\frac{%
u_{l+1}+u_{l}}{2}-a\right)\,,  \label{sdDPf2} \\
&&\frac{ d\,\delta _{l}}{d\,s}= \left( 1- \frac{\delta _{l}^{2}}{4} \right)(u_{l+1}-u_{l})\,,
\label{sdDPf3}
\end{eqnarray}%
where an intermediate variable $r_l$ is used.
\begin{remark}
Due to the fact
$$
\frac{\delta_l}{2b}=\frac{x_{l+1}-x_l}{2b} \to \frac{\partial x}{\partial y}=-\frac{1}{ar}\,,
$$
as $b \to 0$, it is obvious that Eqs. (\ref{sdDPf1}) (or \ref{sdDP5n}) and (\ref{sdDPf2}) converge to Eqs. (\ref{DP5p}) and
(\ref{DP6p}), respectively.
\end{remark}
In order to eliminate the intermediate variable $r_l$, we substitute Eq. (\ref{sdDPf2}) into Eq. (\ref{sdDPf1}) and get
\begin{equation}\label{sd-DP6n}
\fl \frac{u_{l}-u_{l-1}}{\delta _{l-1}}+\frac{\delta _{l-1}}{2}\left(\frac{%
u_{l}+u_{l-1}}{2}-a\right)-\frac{u_{l+1}-u_{l}}{\delta _{l}}+\frac{\delta _{l}}{2}\left(\frac{%
u_{l+1}+u_{l}}{2}-a\right)=\frac{2br_{l}}{r_{l+1}}\frac{\frac{4}{\delta _{l}^{2}}-1%
}{\frac{a^{2}}{b^{2}}-1}\,.
\end{equation}
Defining
\begin{equation}\label{m-def}
\fl m_l = \frac{2}{\delta_{l}+\delta_{l-1}} \left( -\frac{u_{l+1}-u_{l}}{\delta_l} + \frac{u_{l}-u_{l-1}}{\delta_{l-1}}+
\frac{\delta_{l}(u_{l+1}+u_{l})+\delta_{l-1}(u_{l}+u_{l-1})}{4} \right) -a\,,
\end{equation}
and taking the logarithmic derivative on both sides of (\ref{sd-DP6n}), we have
\begin{eqnarray}
\fl   \frac{d\,\ln m_l}{d\,s}  &=& (\ln r_l)_s - (\ln r_{l+1})_s -\frac {8}{(4-\delta_l^2)\delta_l} \frac{d \delta_l}{d s}
  - \frac{d}{d s} \ln(\delta_l+\delta_{l-1}) \nonumber \\
\fl   & = & (\ln r_l)_s - (\ln r_{l+1})_s -2\frac{u_{l+1}-u_{l}}{\delta _{l}}-\frac{d}{d s} \ln(\delta_l+\delta_{l-1}) \qquad    \mbox{(by (\ref{sdDPf3}))}\nonumber \\
\fl  & = & -\frac{u_{l}-u_{l-1}}{\delta _{l-1}}-\frac{\delta _{l-1}}{2}\left(\frac{%
u_{l}+u_{l-1}}{2}-a\right)+\frac{u_{l+1}-u_{l}}{\delta _{l}}+\frac{\delta _{l}}{2}\left(%
\frac{u_{l+1}+u_{l}}{2}-a\right)   \nonumber \\
\fl && \quad  -2\frac{u_{l+1}-u_{l}}{\delta _{l}}-\frac{d}{d s} \ln(\delta_l+\delta_{l-1}) \qquad \mbox{(by (\ref{sdDPf2}))} \nonumber \\
 \fl  & = & -\frac{u_{l}-u_{l-1}}{\delta _{l-1}}-\frac{\delta _{l-1}}{2}\left(\frac{%
u_{l}+u_{l-1}}{2}-a\right)-\frac{u_{l+1}-u_{l}}{\delta _{l}}+\frac{\delta _{l}}{2}\left(%
\frac{u_{l+1}+u_{l}}{2}-a\right)  \nonumber \\
\fl && \quad  -\frac{1}{\delta_l+\delta_{l-1}} \left((u_{l+1}-u_{l-1}) -\frac{\delta^2_l(u_{l+1}-u_{l})+\delta^2_{l-1}(u_{l}-u_{l-1})}{4} \right)\,.
\end{eqnarray}
By defining forward difference and average operators
$$
\Delta u_l = \frac{u_{l+1}-u_l}{\delta_l}, \quad M u_l =\frac{u_{l}+u_{l-1}}{2}\,,
$$
we can summarize what we have deduced into the following theorem.
\begin{theorem}
The semi-discrete Degasperis-Procesi equation
\begin{equation}
\fl \label{sd-DP}
\left\{
\begin{array}{l} \displaystyle
 \frac{d\,\ln m_l}{d \,s}
  = -2M \Delta u_l -\frac{M (\delta_l \Delta u_l)}{M \delta_l} +
  \frac{\delta_l(M u_{l+1}-a)-\delta_{l-1}(M u_{l}-a)}{2}+\frac{M(\delta^2_{l}(u_{l+1}-u_{l}))}{4M \delta_l} \,, \\
\displaystyle \frac{d \, \delta_l}{d\, s} = \left( 1- \frac{\delta _{l}^{2}}{4} \right) (u_{l+1}-u_{l})\, \\
\displaystyle m_l= -\frac{\Delta u_l - \Delta u_{l-1}}{M \delta_l}  + \frac{M(\delta_l (M u_l))}{{M \delta_l}} -a\,,
\end{array}\right.
\end{equation}
is determined from the following equations
\begin{equation*}
\left\{
\begin{array}{l} \displaystyle
    \left(\frac{1}{a+b}D_s-1\right)g_{l+1}\cdot f_l =\left(\frac{1}{a-b}D_s-1\right)g_l\cdot f_{l+1}\,, \\
\displaystyle   g_lf_l=c'G_l\,,\\
\displaystyle (D_s-2b)F_{l+1}\cdot F_l=-2bG^2_l\,, \\
\displaystyle   (a-b)g_{l+1}f_l-(a+b)g_lf_{l+1}=-2bc'F_{l+1}\,.
\end{array}\right.
\end{equation*}
through discrete hodograph transformation
$$\delta _{l}=2\frac{(a-b)g_{l+1}f_{l}-(a+b)g_{l}f_{l+1}}{(a-b)g_{l+1}f_{l}+(a+b)g_{l}f_{l+1}}\,, \quad t=s
$$ and dependent variable transformation
$$u_l= \left(\ln \frac{g_l}{f_l}\right)_{s}$$.
\end{theorem}
Let us consider the continuous limit when $b \to 0$.
The dependent variable $u$ is a function of $l$ and $s$. Meanwhile, we regard
it as a function of $x$ and $t$, where $x$ is the space coordinate
at $l$-th lattice point and $t$ is the time, defined by
$$
x=x_0+\sum_{j=0}^{l-1}\delta_j\,,\qquad t=s
$$
Then in the continuous limit, $b \to 0$ ($\delta_l \to 0$), we have
$$
2M \Delta u_l= \frac{u_{l+1}-u_l}{\delta_l}+ \frac{u_{l}-u_{l-1}}{\delta_{l-1}} \to 2u_x \,,
\quad \frac{M (\delta_l \Delta u_l)}{M \delta_l}= \frac{u_{l+1}-u_{l-1}}{\delta_l+\delta_{l-1}} \to u_x\,,
$$
$$
\frac{\delta_{l-1}}{2}(M u_{l}-a) \to 0 \,,
\quad \frac{M(\delta^2_{l}(u_{l+1}-u_{l}))}{M \delta_l} \to 0\,,
$$
$$
m_l=-\frac{(\Delta u_l - \Delta u_{l-1})}{M \delta_l}  + \frac{M(\delta_l (M u_l))}{{M \delta_l}} -a  \to  m=u- u_{xx}-a\,.
$$
Moreover, since
$$
\frac{\partial x}{\partial s}
=\frac{\partial x_0}{\partial s}
+\sum_{j=0}^{l-1}\frac{\partial\delta_j}{\partial s}
=\frac{\partial x_0}{\partial s}
+\sum_{j=0}^{l-1}  \left( 1- \frac{\delta _{l}^{2}}{4} \right) (u_{j+1}-u_j)
\to u \,,
$$
we then have
$$
\partial_s=\partial_t+\frac{\partial x}{\partial s} \partial_x
\to \partial_t+u\partial_x \,.
$$
Consequently, the third equation in (\ref{sd-DP}) converges to $m=u-u_{xx}-a$. Whereas the first equation in (\ref{sd-DP}) converges to
\begin{equation}
(\partial_t+u\partial_x) m  = -3m u_x\,,
\end{equation}
which is exactly the Degasperis-Procesi equation (\ref{DP-eq}).
Based on the results in previous section, we can provide $N$-soliton solution to the semi-discrete Degasperis-Procesi equation
\begin{theorem}
The $N$-soliton solution to the semi-discrete analogue of the Degasperis-Procesi equation (\ref{sd-DP}) takes the following parametric form
\begin{equation}\label{Nsoliton-DP1}
u_l= \left(\ln \frac{g_l}{f_l}\right)_{s}, \quad   \delta _{l}=2\frac{(a-b)g_{l+1}f_{l}-(a+b)g_{l}f_{l+1}}{(a-b)g_{l+1}f_{l}+(a+b)g_{l}f_{l+1}}\,,
\end{equation}
where $g_l=f_{1l}$, $f_l=f_{0l}$ with pfaffian $f_{kl}$ defined by
\begin{equation}
\label{Nsoliton-DP2}
f_{kl}= {\rm Pf} (1,2,\cdots,2N)_{kl}\,,
\end{equation}
whose elements are
\begin{equation}
\label{Nsoliton-DP3}
(i,j)_{kl}=c_{i,j}
+\frac{p_i-p_j}{p_i+p_j}\varphi_i^{(0)}(k,l)\varphi_j^{(0)}(k,l)\,,
\end{equation}
\begin{equation}
\label{Nsoliton-DP4}
\varphi_i^{(n)}(k,l)
=p_i^n \left(\frac{1+ap_i}{1-ap_i}\right)^k \left(\frac{1+bp_i}{1-bp_i}\right)^le^{p_i^{-1} s+\xi_{i0}}
\end{equation}
under the reduction condition ($i=1, 2, \cdots, N$)
\begin{equation}
\label{Nsoliton-DP5}
\fl p_i (1-a^2p_{i}^2) (1-b^2p_{2N+1-i}^2)+  p_{2N+1-i} (1-a^2p_{2N+1-i}^2) (1-b^2p_i^2)=0\,.
\end{equation}
\end{theorem}
In the last, we  calculate the $\tau$-functions for one- and two-soliton solutions, and compared them with the ones
of the DP equation (\ref{DP-eq}).

\par {\bf One-soliton:} For $N=1$, we have
\begin{eqnarray}
   g_l &=& {\rm Pf}(1,2)_{10}=c_{1,2}+\frac{p_1-p_2}{p_1+p_2} \varphi_1^{(0)}(1,l)\varphi_2^{(0)}(1,l) \\
   & \propto & 1 +  e^{\xi_1(l)+\xi_2(l)+\phi_1}\,,
\end{eqnarray}
\begin{eqnarray}
   f_l &=& {\rm Pf}(1,2)_{00}=c_{1,2}+\frac{p_1-p_2}{p_1+p_2} \varphi_1^{(0)}(0,l)\varphi_2^{(0)}(0,l) \\
   & \propto & 1 +   e^{\xi_1(l)+\xi_2(l)-\phi_1}\,,
\end{eqnarray}
where
\begin{equation*}
  e^{\xi_i(l)}=\left(\frac{1+bp_i}{1-bp_i}\right)^le^{p^{-1}_is+\xi_{i0}}\,, \quad (i=1,2)\,,
  \quad e^{\phi_1}= \sqrt{\frac{(1+ap_1)(1+ap_2)}{(1-ap_1)(1-ap_2)}}\,.
\end{equation*}
Here $p_1$, $p_2=p^*_1$ are two parameters related by a constraint
\begin{equation}
p_1(1-a^2p^2_1)(1-b^2p^2_2)+p_2(1-a^2p^2_2)(1-b^2p^2_1)=0\,.
\label{reduction}
\end{equation}
Let  $p_1=A_1 e^{{\rm i} \theta_1}$, $p_2=A_1 e^{-{\rm i} \theta_1}$, and $p_1+p_2=k_1$, we than have
 \begin{equation}
1-a^2k^2_1 +(3a^2-b^2) A^2_1 +a^2b^2 A^4_1=0\,,
\end{equation}
from which $A^2$ can be solved as
\begin{equation}
A^2_1=\frac{\sqrt{(3a^2-b^2)^2-4a^2b^2(1-a^2k^2_1)}-(3a^2-b^2)}{2a^2b^2}\,.
\end{equation}
In the continuous limit of $b \to 0$ and $a=-1$, a simple calculation gives
\begin{equation}
A^2_1  \to \frac{k^2_1-1}{3}\,,
\end{equation}
and
\begin{equation}
e^{\xi_1(l)+\xi_2(l)} \to e^{-2bl(p_1+p_2) + (p^{-1}_1+p^{-1}_2) s +\eta_{10}}
\to e^{k_1 y + \frac{3k_1}{k^2_1-1} s +\eta_{10}} \,
\end{equation}
by letting $y=-2bl$. Therefore, the one-soliton solution of the semi-discrete DP equation converges to the one-soliton solution given in \cite{Matsuno-DP1,Matsuno-DP2,FMO-DP}.

\par {\bf Two-soliton:}
\noindent {\bf Two-soliton} \\
For $N=2$, by assuming $\varphi_{ij}^{(0)}(k,l) =\varphi_i^{(0)}(k,l)\varphi_j^{(0)}(k,l)$, we have
\begin{eqnarray}
\fl g_l&=&{\rm Pf}(1,2,3,4)_{10}
={\rm Pf}(1,2)_{10}{\rm Pf}(3,4)_{10}-{\rm Pf}(1,3)_{10}{\rm Pf}(2,4)_{10}
+{\rm Pf}(1,4)_{10}{\rm Pf}(2,3)_{10} \nonumber \\
\fl &=&\frac{p_1-p_2}{p_1+p_2}\varphi_{12}^{(0)}(1,l) \times
 \frac{p_3-p_4}{p_3+p_4}\varphi_{34}^{(0)}(1,l)
-\frac{p_1-p_3}{p_1+p_3}\varphi_{13}^{(0)}(1,l) \times
 \frac{p_2-p_4}{p_2+p_4}\varphi_{24}^{(0)}(1,l) \nonumber \\
 \fl &&\qquad +\left(c_{14}+\frac{p_1-p_4}{p_1+p_4} \varphi_{14}^{(0)}(1,l) \right)
\left(c_{23}+\frac{p_2-p_3}{p_2+p_3}\varphi_{23}^{(0)}(1,l)\right) \nonumber \\
 \fl &\propto & 1+e^{\xi_1(l)+\xi_4(l)+\phi_1+\gamma_1}
+e^{\xi_2(l)+\xi_3(l)+\phi_2+ \gamma_2}
+b_{12}e^{\sum^4_{j=1}\xi_j(l)+\phi_1+\phi_2+\gamma_1+\gamma_2}\,,
\end{eqnarray}
\begin{eqnarray}
\fl f_l&=&{\rm Pf}(1,2,3,4)_{00} \nonumber \\
 \fl &\propto & 1+e^{\xi_1(l)+\xi_4(l)-\phi_1+\gamma_1}
+e^{\xi_2(l)+\xi_3(l)-\phi_2+ \gamma_2}
+b_{12}e^{\sum^4_{j=1}\xi_j(l)-\phi_1-\phi_2+\gamma_1+\gamma_2}\,,
\end{eqnarray}
under the condition
\begin{equation}
p_i(1-a^2p^2_i)(1-b^2p^2_{2N+1-i})+p_{2N+1-i}(1-a^2p^2_{2N+1-i})(1-b^2p^2_i)=0\, \quad (i=1,2)\,.
\end{equation}
Here $c_{14}=c_{23}=1$, $e^{\phi_1}= \sqrt{\frac{(1+ap_1)(1+ap_4)}{(1-ap_1)(1-ap_4)}}$,
$e^{\phi_2}= \sqrt{\frac{(1+ap_2)(1+ap_3)}{(1-ap_2)(1-ap_3)}}$,
$e^{\gamma_1}=\frac{p_1-p_4}{p_1+p_4}$, $e^{\gamma_2}=\frac{p_2-p_3}{p_2+p_3}$,
and $b_{12}=\frac{(p_1-p_2)(p_1-p_3)(p_4-p_2)(p_4-p_3)}{(p_1+p_2)(p_1+p_3)(p_4+p_2)(p_4+p_3)}$.
In the continuous limit $b \to 0$, we can show that the two-soliton solutions for semi-discrete DP equation converge to the two-soliton solutions of the DP equation found in  \cite{Matsuno-DP1,Matsuno-DP2,FMO-DP} by letting $p_1+p_4=k_1$, $p_2+p_3=k_2$.
\section{Conclusion and further topics}
In the present paper, we firstly review a set of equations which drive the DP equation through a dependent variable transformation
and a hodograph transformation. Then by constructing the integrable semi-discrete analogues of these equations including bilinear equations and by defining a discrete hodograph transformation, an integrable semi-discrete DP equation was proposed.

Similar to what we have done for the CH equation \cite{dCH2}, the short pulse equation \cite{discreteSP} and a coupled short pulse
equation \cite{discreteCSP}, it deserves exploring a problem of using the proposed semi-discrete DP equation as a self-adaptive moving mesh scheme for the numerical simulation of the DP equation.
\section*{Acknowledgment}
BF appreciated the partial support by the National Natural Science Foundation of China (No. 11428102). The work of KM is partially supported by CREST, JST. The work of YO is partly supported by JSPS Grant-in-Aid for Scientific Research (B-24340029, C-15K04909) and for Challenging Exploratory Research (26610029).

\section*{Appendix A}
\renewcommand{\theequation}{A.\arabic{equation}}
\setcounter{equation}{0}
\begin{eqnarray}
  \partial_s f_{kl} &=& \partial_s \mathop{\rm Pf}_{1\le i<j\le 2N}
\Big((i,j)_{kl}\Big) \nonumber  \\
  &=& {\rm Pf}\pmatrix{
\matrix{(i,j)_{kl} \cr {}} &\varphi_i^{(-1)}(k,l) &\varphi_i^{(0)}(k,l) \cr
&&0} \nonumber \\
  &=& (1,2,\cdots,2N,d_{-1},d_0)_{kl}
\end{eqnarray}
\begin{eqnarray}
\fl f_{k+1,l}&=&\mathop{\rm Pf}_{1\le i<j\le 2N}\Big((i,j)_{kl}
+\varphi_i^{(0)}(k+1,l)\varphi_j^{(0)}(k,l)
-\varphi_i^{(0)}(k,l)\varphi_j^{(0)}(k+1,l)\Big) \nonumber
\\
\fl &=&{\rm Pf}\pmatrix{
\matrix{(i,j)_{kl} \cr {}} &\varphi_i^{(0)}(k,l) &\varphi_i^{(0)}(k+1,l) \cr
&&1} \nonumber \\
\fl &=&(1,2,\cdots,2N,d_0,d^k)_{kl}
\end{eqnarray}
\begin{eqnarray}
\fl (\partial_s-a)f_{k+1,l}
&=&(\partial_s-a){\rm Pf}\pmatrix{
\matrix{(i,j)_{kl} \cr {}} &\varphi_i^{(0)}(k,l) &\varphi_i^{(0)}(k+1,l) \cr
&&1} \nonumber
\fl \\&=&{\rm Pf}\pmatrix{
\matrix{(i,j)_{kl} \cr {}} &\varphi_i^{(-1)}(k,l) &\varphi_i^{(0)}(k,l)
&\varphi_i^{(0)}(k,l) &\varphi_i^{(0)}(k+1,l) \cr
&&0&0&0 \cr
&&&0&0 \cr
&&&&1} \nonumber \\
\fl && \ \ +{\rm Pf}\pmatrix{
\matrix{(i,j)_{kl} \cr {}} &\partial_s\varphi_i^{(0)}(k,l)
&\varphi_i^{(0)}(k+1,l) \cr
&&0} \nonumber \\
&& \ \ +{\rm Pf}\pmatrix{
\matrix{(i,j)_{kl} \cr {}} &\varphi_i^{(0)}(k,l)
&(\partial_s-a)\varphi_i^{(0)}(k+1,l) \cr
&&-a} \nonumber \\
\fl &=&{\rm Pf}\pmatrix{
\matrix{(i,j)_{kl} \cr {}} &\varphi_i^{(-1)}(k,l) &\varphi_i^{(0)}(k,l)
&0 &\varphi_i^{(0)}(k+1,l) \cr
&&0&0&0 \cr
&&&0&0 \cr
&&&&1} \nonumber \\
\fl && \ \ +{\rm Pf}\pmatrix{
\matrix{(i,j)_{kl} \cr {}} &\varphi_i^{(-1)}(k,l)
&\varphi_i^{(0)}(k+1,l) \cr
&&0} \nonumber \\
&& \ \ +{\rm Pf}\pmatrix{
\matrix{(i,j)_{kl} \cr {}} &\varphi_i^{(0)}(k,l)
&\varphi_i^{(-1)}(k,l)+a\varphi_i^{(0)}(k,l) \cr
&&-a} \nonumber \\
\fl &=& {\rm Pf}\pmatrix{
\matrix{(i,j)_{kl} \cr {}} &\varphi_i^{(-1)}(k,l) &\varphi_i^{(0)}(k,l) \cr
&&0} \nonumber \\
&& \  \ +{\rm Pf}\pmatrix{
\matrix{(i,j)_{kl} \cr {}} &\varphi_i^{(-1)}(k,l)
&\varphi_i^{(0)}(k+1,l) \cr
&&0}
+{\rm Pf}\pmatrix{
\matrix{(i,j)_{kl} \cr {}} &\varphi_i^{(0)}(k,l)
&\varphi_i^{(-1)}(k,l) \cr
&&-a} \nonumber \\
\fl &=&{\rm Pf}\pmatrix{
\matrix{(i,j)_{kl} \cr {}} &\varphi_i^{(-1)}(k,l)
&\varphi_i^{(0)}(k+1,l) \cr
&&-a} \nonumber \\
\fl &=&(1,2,\cdots,2N,d_{-1},d^k)_{kl}
\end{eqnarray}
Similarly, we have
\begin{equation}
  f_{k,l+1}=(1,2,\cdots,2N,d_0,d^l)_{kl}\,.
\end{equation}
\begin{equation}
 (\partial_s-b)f_{k,l+1}=(1,2,\cdots,2N,d_{-1},d^l)_{kl}\,.
\end{equation}

\begin{eqnarray}
\fl &&\frac{a-b}{a+b}f_{k+1,l+1}
=\frac{a-b}{a+b}{\rm Pf}\pmatrix{
\matrix{(i,j)_{k,l+1} \cr {}} &\varphi_i^{(0)}(k,l+1)
&\varphi_i^{(0)}(k+1,l+1) \cr
&&1}
 \nonumber \\
\fl  &&\qquad
=\frac{a-b}{a+b}{\rm Pf}\pmatrix{
\matrix{(i,j)_{kl} \cr {}} &\varphi_i^{(0)}(k,l) &\varphi_i^{(0)}(k,l+1)
&\varphi_i^{(0)}(k,l+1) &\varphi_i^{(0)}(k+1,l+1) \cr
&&1&0&0 \cr
&&&0&0 \cr
&&&&1}
 \nonumber \\ \fl &&\qquad
=\frac{a-b}{a+b}{\rm Pf}\pmatrix{
\matrix{(i,j)_{kl} \cr {}} &\varphi_i^{(0)}(k,l) &\varphi_i^{(0)}(k,l+1)
&\varphi_i^{(0)}(k,l+1) &\varphi_i^{(0)}(k+1,l+1)-\varphi_i^{(0)}(k,l) \cr
&&1&0&0 \cr
&&&0&1 \cr
&&&&1}
 \nonumber \\ \fl &&\qquad
=\frac{a-b}{a+b}{\rm Pf}\pmatrix{
\matrix{(i,j)_{kl} \cr {}} &\varphi_i^{(0)}(k,l) &0
&\varphi_i^{(0)}(k,l+1) &\varphi_i^{(0)}(k+1,l+1)-\varphi_i^{(0)}(k,l) \cr
&&1&0&0 \cr
&&&0&0 \cr
&&&&1}
 \nonumber \\ \fl &&\qquad
=\frac{a-b}{a+b}{\rm Pf}\pmatrix{
\matrix{(i,j)_{kl} \cr {}} &\varphi_i^{(0)}(k,l+1)
&\varphi_i^{(0)}(k+1,l+1)-\varphi_i^{(0)}(k,l) \cr
&&1}
 \nonumber \\ \fl &&\qquad
={\rm Pf}\pmatrix{
\matrix{(i,j)_{kl} \cr {}} &\varphi_i^{(0)}(k,l+1)
&\displaystyle\frac{a-b}{a+b}
(\varphi_i^{(0)}(k+1,l+1)-\varphi_i^{(0)}(k,l)) \cr
&&\displaystyle\frac{a-b}{a+b}}
 \nonumber \\ \fl &&\qquad
={\rm Pf}\pmatrix{
\matrix{(i,j)_{kl} \cr {}} &\varphi_i^{(0)}(k,l+1)
&\varphi_i^{(0)}(k+1,l)-\varphi_i^{(0)}(k,l+1) \cr
&&\displaystyle\frac{a-b}{a+b}}
\nonumber \\ \fl &&\qquad
={\rm Pf}\pmatrix{
\matrix{(i,j)_{kl} \cr {}} &\varphi_i^{(0)}(k,l+1)
&\varphi_i^{(0)}(k+1,l) \cr
&&\displaystyle\frac{a-b}{a+b}} \nonumber \\
\fl && \qquad =(1,2,\cdots,2N,d^l,d^k)_{kl}
\end{eqnarray}

\begin{eqnarray}
\fl &&(\partial_s-a-b)\frac{a-b}{a+b}f_{k+1,l+1}
=(\partial_s-a-b){\rm Pf}\pmatrix{
\matrix{(i,j)_{kl} \cr {}} &\varphi_i^{(0)}(k,l+1)
&\varphi_i^{(0)}(k+1,l) \cr
&&\displaystyle\frac{a-b}{a+b}}
\nonumber \\
\fl &&\qquad
 ={\rm Pf}\pmatrix{
\matrix{(i,j)_{kl} \cr {}} &\varphi_i^{(-1)}(k,l) &\varphi_i^{(0)}(k,l)
&\varphi_i^{(0)}(k,l+1) &\varphi_i^{(0)}(k+1,l) \cr
&&0&0&0 \cr
&&&0&0 \cr
&&&&\displaystyle\frac{a-b}{a+b}}
\nonumber \\
\fl &&\qquad
+{\rm Pf}\pmatrix{
\matrix{(i,j)_{kl} \cr {}} &(\partial_s-b)\varphi_i^{(0)}(k,l+1)
&\varphi_i^{(0)}(k+1,l) \cr
&&\displaystyle -b\frac{a-b}{a+b}}
\nonumber \\
\fl &&\qquad
+{\rm Pf}\pmatrix{
\matrix{(i,j)_{kl} \cr {}} &\varphi_i^{(0)}(k,l+1)
&(\partial_s-a)\varphi_i^{(0)}(k+1,l) \cr
&&\displaystyle -a\frac{a-b}{a+b}}
\nonumber \\
\fl &&\qquad
={\rm Pf}\pmatrix{
\matrix{(i,j)_{kl} \cr {}} &\varphi_i^{(-1)}(k,l) &\varphi_i^{(0)}(k,l)
&\varphi_i^{(0)}(k,l+1) &\varphi_i^{(0)}(k+1,l) \cr
&&0&0&0 \cr
&&&0&0 \cr
&&&&\displaystyle\frac{a-b}{a+b}}
\nonumber \\&&\qquad
+{\rm Pf}\pmatrix{
\matrix{(i,j)_{kl} \cr {}} &\varphi_i^{(-1)}(k,l)+b\varphi_i^{(0)}(k,l)
&\varphi_i^{(0)}(k+1,l) \cr
&&\displaystyle -b\frac{a-b}{a+b}}
\nonumber \\
\fl &&\qquad
+{\rm Pf}\pmatrix{
\matrix{(i,j)_{kl} \cr {}} &\varphi_i^{(0)}(k,l+1)
&\varphi_i^{(-1)}(k,l)+a\varphi_i^{(0)}(k,l) \cr
&&\displaystyle -a\frac{a-b}{a+b}}
\nonumber \\
\fl &&\qquad
={\rm Pf}\pmatrix{
\matrix{(i,j)_{kl} \cr {}} &\varphi_i^{(-1)}(k,l) &\varphi_i^{(0)}(k,l)
&\varphi_i^{(0)}(k,l+1) &\varphi_i^{(0)}(k+1,l) \cr
&&0&0&-a \cr
&&&0&1 \cr
&&&&\displaystyle\frac{a-b}{a+b}}
\nonumber \\
\fl &&\qquad
+{\rm Pf}\pmatrix{
\matrix{(i,j)_{kl} \cr {}} &\varphi_i^{(-1)}(k,l)+b\varphi_i^{(0)}(k,l)
&\varphi_i^{(0)}(k+1,l) \cr
&&\displaystyle -a\frac{a-b}{a+b}-b\frac{a-b}{a+b}}
\nonumber \\
\fl &&\qquad
={\rm Pf}\pmatrix{
\matrix{(i,j)_{kl} \cr {}} &\varphi_i^{(-1)}(k,l) &\varphi_i^{(0)}(k,l)
&\varphi_i^{(0)}(k,l+1) &\varphi_i^{(0)}(k+1,l) \cr
&&0&0&-a \cr
&&&0&1 \cr
&&&&\displaystyle\frac{a-b}{a+b}}
\nonumber \\&&\qquad
+{\rm Pf}\pmatrix{
\matrix{(i,j)_{kl} \cr {}} &\varphi_i^{(-1)}(k,l)+b\varphi_i^{(0)}(k,l)
&\varphi_i^{(0)}(k+1,l) \cr
&&-a+b} \nonumber \\
\fl &&\qquad
={\rm Pf}\pmatrix{
\matrix{(i,j)_{kl} \cr {}} &\varphi_i^{(-1)}(k,l) &\varphi_i^{(0)}(k,l)
&\varphi_i^{(0)}(k,l+1) &\varphi_i^{(0)}(k+1,l) \cr
&&0&-b&-a \cr
&&&1&1 \cr
&&&&\displaystyle\frac{a-b}{a+b}}
 \nonumber \\&&\qquad
=(1,2,\cdots,2N,d_{-1},d_0,d^l,d^k)_{kl}
\end{eqnarray}

\end{document}